\documentclass[11pt]{article}

\usepackage{fullpage}
\usepackage{amsmath}
\usepackage{amsthm}
\usepackage{amssymb}
\usepackage{amsfonts}
\usepackage{hyperref}
\usepackage{color}
\usepackage{xcolor}
\usepackage{mathbbol}
\usepackage{algpseudocode}
\usepackage{multicol}
\usepackage{relsize}

\usepackage{mathabx}

\newcommand{\sac}[1]{%
  \begingroup
  \begingroup\lccode`~=`, \lowercase{\endgroup
    \edef~{\mathchar\the\mathcode`, \penalty0 \noexpand\hspace{0pt plus 1em}}%
  }\mathcode`,="8000 #1%
  \endgroup
}

\DeclareMathOperator{\lcm}{lcm}

\newcommand{\ignore}[1]{}

\newtheorem{theorem}{Theorem}

\newtheorem{lemma}[theorem]{Lemma}

\renewcommand{\Pr}{{\bf Pr}}

\newcommand{\D}{{\cal D}}

\newcommand{\N}{{\mathbb{N}}}

\newcommand{\sLinear}{{\mbox{-Linear}}}
\newcommand{\Linear}{{\mbox{Linear}}}

\begin{document}

\title{An Optimal Tester for $k$-Linear}
\author{{\bf Nader H. Bshouty}\\ Dept. of Computer Science\\ Technion,  Haifa, 32000\\
}

\maketitle
\begin{abstract}
A Boolean function $f:\{0,1\}^n\to \{0,1\}$ is $k$-linear if it returns the sum (over the binary field $F_2$) of $k$ coordinates of the input. In this paper, we study property testing of the classes $k$-Linear, the class of all $k$-linear functions, and $k$-Linear$^*$, the class $\cup_{j=0}^kj$-Linear.
We give a non-adaptive distribution-free two-sided $\epsilon$-tester for $k$-Linear that makes
$$O\left(k\log k+\frac{1}{\epsilon}\right)$$ queries.
This matches the lower bound known from the literature.

We then give a non-adaptive distribution-free one-sided $\epsilon$-tester for $k$-Linear$^*$ that makes the same number of queries and show that any non-adaptive uniform-distribution one-sided $\epsilon$-tester for $k$-Linear must make at least $ \tilde\Omega(k)\log n+\Omega(1/\epsilon)$ queries. The latter bound, almost matches the upper bound $O(k\log n+1/\epsilon)$ known from the literature. We then show that any adaptive uniform-distribution one-sided $\epsilon$-tester for $k$-Linear must make at least $\tilde\Omega(\sqrt{k})\log n+\Omega(1/\epsilon)$ queries.
\end{abstract}


\section{Inroduction}

Property testing of Boolean function was first considered in the seminal works of Blum, Luby and Rubinfeld~\cite{BlumLR93} and Rubinfeld and Sudan~\cite{RubinfeldS96} and has recently become a very active research area. See for example,~\cite{AlonKKLR05,BaleshzarMPR16,BelovsB16,BhattacharyyaKSSZ10,BlaisBM11,BlaisK12,Bshouty19,ChakrabartyS13,ChakrabartyS16a,ChakrabortyGM11,ChenDST15,ChenST14,ChenWX17,ChenWX17b,DiakonikolasLMORSW07,FischerKRSS02,
GoldreichGLRS00,GopalanOSSW11,KhotMS15,KhotS16,MatulefORS10,MatulefORS09,ParnasRS02,Saglam18} and other works referenced in the surveys and books~\cite{GoldreichSurvey10,Goldreich17,Ron08,Ron09}.

A Boolean function $f:\{0,1\}^n\to \{0,1\}$ is said to be linear if it returns the sum (over the binary field $F_2$) of some coordinates of the input, $k$-linear if it returns the sum of $k$ coordinates, and, $k$-linear$^*$ if it returns the sum of at most $k$ coordinates. The class Linear (resp. $k$-Linear and $k\sLinear^*$) is the classes of all linear functions (resp. all $k$-linear functions and $\cup_{i=0}^k k\sLinear$). Those classes has been of particular interest to the property testing community~\cite{BlaisBM11,BlaisK12,BlumLR93,Bshouty19,BuhrmanGMW13,FischerKRSS02,Goldreich10,
Goldreich17,HalevyK07,Ron08,Ron09,RubinfeldS11,Saglam18}.

\subsection{The Model}

Let $f$ and $g$ be two Boolean functions $\{0,1\}^n\to \{0,1\}$ and let $\D$ be a distribution on $\{0,1\}^n$. We say that $f$ is {\it $\epsilon$-far} from $g$ with respect to (w.r.t.) $\D$ if $\Pr_\D[f(x)\not= g(x)]\ge \epsilon$ and {\it $\epsilon$-close} to $g$ w.r.t. $\D$ if $\Pr_\D[f(x)\not= g(x)]\le \epsilon$.

In the uniform-distribution and distribution-free property testing model, we consider the problem of testing a class of Boolean function $C$. In the distribution-free testing model (resp. uniform-distribution testing model), the {\it tester} is a randomized algorithm that has access to a Boolean function $f:\{0,1\}^n\to\{0,1\}$ via a black-box oracle that returns $f(x)$ when a string $x$ is queried. The tester also has access to unknown distribution $\D$ (resp. uniform distribution) via an oracle that returns $x\in\{0,1\}^n$ chosen randomly according to the distribution~$\D$ (resp. according to the uniform distribution).
A {\it distribution-free tester},~\cite{GoldreichGR98}, (resp. {\it uniform-distribution tester}) ${\cal A}$ for $C$ is an tester that, given as input a distance parameter $\epsilon$ and the above two oracles to a Boolean function $f$,
\begin{enumerate}
\item if $f\in C$ then ${\cal A}$ accepts with probability at least $2/3$.
\item if $f$ is $\epsilon$-far from every $g\in C$ w.r.t. $\D$ (resp. uniform distribution) then ${\cal A}$ rejects with probability at least $2/3$.
\end{enumerate}

We will also call ${\cal A}$ an {\it $\epsilon$-tester for the class $C$} or an algorithm for {\it $\epsilon$-testing $C$}. We say that ${\cal A}$ is {\it one-sided} if it always accepts when $f\in C$; otherwise, it is called {\it two-sided} tester. The {\it query complexity of ${\cal A}$} is the maximum number of queries ${\cal A}$ makes on any Boolean function~$f$. If the query complexity is $q$ then we call the tester a {\it $q$-query tester} or a tester with {\it query complexity $q$}.

In the {\it adaptive testing} (uniform-distribution or distribution-free) the queries can depend on the answers of the previous queries where in the {\it non-adaptive testing} all the queries are fixed in advance by the tester.

\ignore{In the {\it non-adaptive distribution-free testing} the tester first
gets a labeled sample $(x^{(1)}, f(x^{(1)})), . . . , $ $(x^{(q')}, f(x^{(q')}))$ from the sample oracle. Based on them, it asks black box queries on strings $y^{(1)}, . . . , y^{(q)}$. The $y^{(i)}$’s may depend on the random pairs $(x^{(j)}, f(x^{(j)}))$ received from the sampling oracle, but the $i$-th black-box query string $y^{(i)}$ may not
depend on the responses $f(y^{(1)}), . . . , f(y^{(i-1)})$ to any of the $i-1$ earlier queries.}

In this paper we study testers for the classes $k$-Linear and $k$-Linear$^*$.

\subsection{Prior Results}
Throughout this paper we assume that $k<\sqrt{n}$. Blum et al. \cite{BlumLR93} gave an $O(1/\epsilon)$-query non-adaptive uniform-distribution one-sided $\epsilon$-tester (called BLR tester) for Linear. Halevy and Kushilevitz,~\cite{HalevyK07}, used a self-corrector (an algorithm that computes $g(x)$ from a black box query to $f$ that is $\epsilon$-close to $g$) to reduce distribution-free testability to uniform-distribution testability. This reduction gives an $O(1/\epsilon)$-query non-adaptive distribution-free one-sided $\epsilon$-tester for Linear. The reduction can be applied to any subclass of Linear. In particular, any $q$-query uniform-distribution $\epsilon$-tester for $k$-Linear ($k\sLinear^*$) gives a $O(q)$-query distribution-free $\epsilon$-tester.

It is well known that if there is a $q_1$-query uniform-distribution $\epsilon$-tester for Linear and a $q_2$-query uniform-distribution $\epsilon$-tester for the class $k$-Junta\footnote{The class of boolean functions that depends on at most $k$ coordinates} then there is an $O(q_1+q_2)$-query uniform-distribution $O(\epsilon)$-tester for $k\sLinear^*$. Since $k\sLinear=k\sLinear^*\backslash (k-1)\sLinear^*$, if there is a $q$-query uniform-distribution $\epsilon$-tester for $k\sLinear^*$ then there is an $O(q)$-query uniform-distribution two-sided $\epsilon$-tester for $k\sLinear$. Therefore, all the results for testing $k$-Junta are also true for $k\sLinear^*$ and $k$-Linear in the uniform-distribution model.

For lower bounds on the number queries for two-sided uniform-distribution testing $k$-Linear (see the table in Figure~\ref{Table}): For non-adaptive testers
Fisher, et al.~\cite{FischerKRSS02} gave the lower bound $\Omega(\sqrt{k})$. Goldreich~\cite{Goldreich10}, gave the lower bound $\Omega(k)$. In~\cite{BlaisK12}, Blais and Kane gave the lower bound $2k-o(k)$. Then in~\cite{BlaisBM11}, Blais et al. gave the lower bound $\Omega(k\log k)$. For adaptive testers, Goldreich~\cite{Goldreich10}, gave the lower bound $\Omega(\sqrt{k})$. Then Blais et al.~\cite{BlaisBM11} gave the lower bound $\Omega(k)$ and in~\cite{BlaisK12}, Blais and Kane gave the lower bound $k-o(k)$. Then in~\cite{Saglam18}, Saglam gave the lower bound $\Omega(k\log k)$. This bound with the trivial $\Omega(1/\epsilon)$ lower bound gives the lower bound
\begin{eqnarray}\label{Lower}
\Omega\left(k\log k+\frac{1}{\epsilon}\right)
\end{eqnarray}
for the query complexity of any adaptive uniform-distribution (and distribution-free) two-sided testers.

For upper bounds for uniform-distribution two-sided $\epsilon$-testing $k$-Linear, Fisher, et al.~\cite{FischerKRSS02} gave the first adaptive tester that makes $O(k^2/\epsilon)$ queries. In \cite{BuhrmanGMW13}, Buhrman et al. gave a non-adaptive tester that makes $O(k\log k)$ queries for any constant~$\epsilon$.
As is mentioned above, testing $k$-Linear can be done by first testing if the function is $k$-Junta and then testing if it is Linear. Therefore, using Blais~\cite{Blais08,Blais09} adaptive and non-adaptive testers for $k$-Junta we get adaptive and non-adaptive uniform-distribution testers for $k$-Linear that makes $ O(k\log k+k/\epsilon)$ and $\tilde O(k^{1.5}/\epsilon)$ queries, respectively.

For upper bounds for two-sided distribution-free testing $k$-Linear, as is mentioned above, from Halevy et al. reduction in~\cite{HalevyK07}, an adaptive and non-adaptive distribution-free $\epsilon$-tester can be constructed from adaptive and non-adaptive uniform-distribution $\epsilon$-testers. This gives an adaptive and non-adaptive distribution-free two-sided testers for $k$-Linear that makes  $O(k\log k+k/\epsilon)$ and $\tilde O(k^{1.5}/\epsilon)$ queries, respectively. See the table in Figure~\ref{Table}.

\subsection{Our Results}
In this paper we prove
\begin{theorem} For any $\epsilon>0$, there is a polynomial time non-adaptive distribution-free one-sided $\epsilon$-tester for $k$-Linear$^*$ that makes $$O\left(k\log k+\frac{1}{\epsilon}\right)$$ queries.
\end{theorem}

By the reduction from $k\sLinear$ to $k\sLinear^*$, we get
\begin{theorem} For any $\epsilon>0$, there is a polynomial time non-adaptive distribution-free two-sided $\epsilon$-tester for $k$-Linear that makes $$O\left(k\log k+\frac{1}{\epsilon}\right)$$ queries.
\end{theorem}
\ignore{
When the success of the tester is amplified to $1-\delta$ by the standard majority vote, the query complexity becomes $O(k(\log k)\log (1/\delta)+(1/\epsilon)\log(1/\delta))$. In this paper we show

\begin{theorem} For any $\epsilon>0$ and $k^2\le \delta n$, there is a polynomial time non-adaptive distribution-free two-sided $\epsilon$-tester for $k$-Linear that has success probability at least $1-\delta$ and makes $$O\left(k\log \frac{k}{\delta}+\frac{1}{\epsilon}\log\frac{1}{\delta}\right)$$ queries.
\end{theorem}
By Saglam's lower bound in~\cite{Saglam18} and the trivial lower bound $\Omega((1/\epsilon)\log(1/\delta))$, our tester is optimal.}

For one-sided testers for $k$-Linear we prove

\begin{theorem} Any non-adaptive uniform-distribution one-sided $\epsilon$-tester for $k$-Linear must make at least  $\tilde\Omega(k)\log n+\Omega(1/\epsilon)$ queries.
\end{theorem}
This almost matches the upper bound $O(k\log n+1/\epsilon)$ that follows from the reduction of Goldreich et. al~\cite{GoldreichGR98} and the non-adaptive deterministic exact learning algorithm of Hofmeister~\cite{Hofmeister99} that learns $k$-Linear with $O(k\log n)$ queries.

For adaptive testers we prove
\begin{theorem} Any adaptive uniform-distribution one-sided $\epsilon$-tester for $k$-Linear must make at least  $\tilde\Omega(\sqrt{k})\log n+\Omega(1/\epsilon)$ queries.
\end{theorem}

The table in~\ref{Table} summarizes all the results in the literature and our results for the class $k\sLinear$.

\begin{figure}[h!]
\begin{center}
\begin{tabular}{|l|c|l|l|c|c|}
 Upper/ &One-Sided/ & Adaptive/ &Uniform/& &  \\
 Lower &Two-Sided & Non-Adap. &Dist. Free&Result $O/\Omega$ & Reference \\
 \hline
 \hline
Upper &Two-Sided  &Adaptive&Uniform &$k^2/\epsilon$&\cite{FischerKRSS02}\\
\hline
 Upper &Two-Sided & Adaptive &Uniform &$k\log k+k/\epsilon$&\cite{Blais09}\\
  \hline
 Upper &Two-Sided & Adaptive &Dist. Free &$k\log k+k/\epsilon$&\cite{HalevyK07}\\
  \hline
  Upper &Two-Sided & Non-Adap.&Uniform &$k\log k$ ($\epsilon$ Const.)&\cite{BuhrmanGMW13}\\
 \hline
 Upper &Two-Sided& Non-Adap. &Uniform &$k^{1.5}/\epsilon$&\cite{Blais08}\\
\hline
 Upper &Two-Sided & Non-Adap. &Dist. Free &$k^{1.5}/\epsilon$&\cite{HalevyK07}\\
 \hline
 Upper &Two-Sided & Non-Adap. &Dist. Free &$k\log k+1/\epsilon$&Ours\\
\hline\hline
 Lower &Two-Sided &Non-Adap. &Uniform &$1/\epsilon$&Trivial \\
  \hline
 Lower &Two-Sided &Non-Adap. &Uniform &$\sqrt{k}+1/\epsilon$&\cite{FischerKRSS02} \\
 \hline
 Lower &Two-Sided &Non-Adap. &Uniform &${k}+1/\epsilon$&\cite{Goldreich10} \\
  \hline
 Lower &Two-Sided & Non-Adap. &Uniform &${k\log k}+1/\epsilon$&\cite{BlaisBM11} \\
 \hline
 Lower &Two-Sided &Adaptive &Uniform &$\sqrt{k}+1/\epsilon$&\cite{Goldreich10} \\
 \hline
 Lower &Two-Sided & Adaptive &Uniform &${k}+1/\epsilon$&\cite{BlaisBM11,BlaisK12} \\
 \hline
 Lower &Two-Sided & Adaptive &Uniform &${k\log k}+1/\epsilon$&\cite{Saglam18} \\
 \hline\hline
 Upper &One-Sided & Non-Adaptive &Dist. Free &${k\log n}+1/\epsilon$&\cite{GoldreichGR98} \\
 \hline
 Lower &One-Sided & Non-Adaptive &Uniform &$\tilde\Omega(k){\log n}+1/\epsilon$&Ours \\
 \hline
 Lower &One-Sided & Adaptive &Uniform &$\tilde\Omega(\sqrt{k}){\log n}+1/\epsilon$&Ours \\
 \hline
\end{tabular}
\end{center}
	\caption{A table of results for the testability of the class $k$-Linear.}
	\label{Table}
\end{figure}

\section{Overview of the Testers and Lower Bounds}
In this section we give overview of the techniques used for proving the results in this paper.

\subsection{One-sided Tester for $k$-Linear$^*$}
The tester for $k$-Linear$^*$ first runs the tester BLR of Blum et al.~\cite{BlumLR93} to test if the function $f$ is $\epsilon'$-close to Linear w.r.t. the uniform distribution, where $\epsilon'=\Theta(1/(k\log k))$. BLR is one-sided tester and therefore, if $f$ is $k$-linear then BRG accepts with probability 1. If $f$ is $\epsilon'$-far from Linear w.r.t. the uniform distribution then, with probability at least $2/3$, BLR rejects. Therefore, if the tester BLR accepts, we may assume that $f$ is $\epsilon'$-close to Linear w.r.t. the uniform distribution. Let $g\in$Linear be the function that is $\epsilon'$-close to $f$. If $f$ is $k$-linear$^*$ then $f=g$. This is because $\epsilon'<1/8$ and the distance (w.r.t. the uniform distribution) between every two linear functions is $1/2$. BLR makes $O(1/\epsilon')=O(k\log k)$ queries.

In the second stage, the tester tests if $g$ (not $f$) is $k$-linear$^*$. Let us assume for now that we can query $g$ in every string. Since $g\in$Linear, we need to distinguish between functions in $k$-Linear$^*$ and functions in Linear$\backslash k$-Linear$^*$.
We do that with two tests. We first test if $g\in  8k$-Linear$^*$  and then test if it is in $k$-Linear$^*$ assuming that it is in $8k$-Linear${}^*$. In the first test, the tester ``throws'', uniformly at random, the variables of $g$ into $16k$ bins and tests if there is more than $k$ non-empty bins. If $g$ is $k$-linear$^*$ then the number of non-empty bins is always less than $k$. If it is $k'$-linear for some $k'>8k$ then with high probability (w.h.p.) the number of non-empty bins is greater than $k$. Notice that if $f$ is $k$-linear$^*$ then the test always accepts and therefore it is one-sided. This tests makes $O(k)$ queries to $g$.

The second test is testing if $g$ is in $k$-Linear$^*$ assuming that it is in $8k$-Linear${}^*$. This is done by projecting the variables of $g$ into $r=O(k^2)$ coordinates uniformly at random and learning (finding exactly) the projected function using the non-adaptive deterministic Hofmeister's algorithm,~\cite{Hofmeister99}, that makes $O(k\log r)=O(k\log k)$ queries. Since $g\in 8k$-Linear${}^*$, w.h.p., the relevant coordinates of the function are projected to different coordinates, and therefore, w.h.p., the learning gives a linear function that has exactly the same number of relevant coordinates as $g$. The tester accepts if the number of relevant coordinates in the projected function is at most $k$.
If $g\in k$-Linear$^*$, then the projected function is in $k$-Linear$^*$ with probability~1 and therefore this test is one-sided. This test makes $O(k\log k)$ queries.

We assumed that we can query $g$. We now show how to query $g$ in $O(k\log k)$ strings so we can apply the above two tests. For this, the tester uses self-corrector,~\cite{BlumLR93}. To compute $g(z)$, the self-corrector chooses a uniform random string $a\in\{0,1\}^n$ and computes $f(z+a)+f(a)$. Since $f$ is $O(1/(k\log k))$-close to $g$ w.r.t. the uniform distribution, we have that for any string $z\in\{0,1\}^n$ and an $a\in\{0,1\}^n$ chosen uniformly at random, with probability at least $1-O(1/(k\log k))$,  $f(z+a)+f(a)=g(z+a)+g(a)=g(z)$. Therefore, w.h.p., the self-corrector computes correctly the values of $g$ in $O(k\log k)$ strings. If $f\in k$-Linear then $g=f$ and $f(z+a)+f(z)=f(z)=g(z)$, i.e., the self-corrector gives the value of $g$ with probability $1$. This shows that the above two tests are one-sided.

Now, if $f$ is $k$-linear$^*$ then $f=g$. If $f$ is $\epsilon$-far from every function in $k$-Linear$^*$ w.r.t. $\D$ then it is $\epsilon$-far from $g$ w.r.t.~$\D$.

In the final stage the tester tests whether $f$ is equal to $g$ or $\epsilon$-far from $g$ w.r.t. $\D$. Here again the tester uses self-corrector. It asks for a sample $\{(z^{(i)},f(z_i))|i\in [t]\}$ according to the distribution $\D$ of size $t=O(1/\epsilon)$ and tests if $f(z^{(i)})=f(z^{(i)}+a^{(i)})+f(a^{(i)})$ for every $i\in [t]$, where $a^{(i)}$ are i.i.d. uniform random strings. If $f(z^{(i)})=f(z^{(i)}+a^{(i)})+f(a^{(i)})$ for all $i$ then it accepts, otherwise, it rejects. If $f$ is $k$-linear then $f(z^{(i)})=f(z^{(i)}+a^{(i)})+f(a^{(i)})$ for all $i$ and the tester accepts with probability $1$. Now suppose $f$ is $\epsilon$-far from $g$ w.r.t. $\D$. Since $f$ is $\epsilon'$-close to $g$ w.r.t. the uniform distribution and $\epsilon'\le 1/8$ we have that, with probability at least $7/8$,  $f(z^{(i)}+a^{(i)})+f(a^{(i)})=g(z^{(i)}+a^{(i)})+g(a^{(i)})=g(z^{(i)})$. Therefore, assuming the latter happens, then, with probability at least $1-\epsilon$ we have $f(z^{(i)})\not=g(z^{(i)})=f(z^{(i)}+a^{(i)})+f(a^{(i)})$. Thus, w.h.p, there is $i$ such that $f(z^{(i)})\not=f(z^{(i)}+a^{(i)})+f(a^{(i)})$ and the tester rejects. This stage is one-sided and makes $O(1/\epsilon)$ queries.

\subsection{Two-sided Testers for $k$-Linear}

As we mentioned in the introduction, the one-sided $q$-query uniform-distribution $\epsilon$-tester for $k\sLinear^*$ gives a two-sided uniform-distribution $O(q)$-query $\epsilon$-tester for $k\sLinear$. This is because, in the uniform distribution, the linear functions are $1/2$-far from each other and therefore, for any $\epsilon<1/4$, if $f$ is $\epsilon$-close to a $k$-linear function $g$ then it is $(1/2-\epsilon)$-far from $(k-1)$-Linear$^*$. This is not true for any distribution $\D$, and therefore, cannot be applied here.

The algorithm in the previous subsection can be changed to a two-sided tester for $k$-Linear as follows. The only part that should be changed is the test that $g$ is in $k$-Linear$^*$ assuming that it is in $8k$-Linear${}^*$. We replace it with a test that $g$ is in $k$-Linear assuming that it is in $8k$-Linear${}^*$. The tester rejects if the number of relevant coordinates in the function that is learned is not {\it equal} to $k$. This time the test is two-sided. The reason is that the projection to $O(k^2)$ variables does not guarantee (with probability $1$) that all the variables of $f$ are projected to different variables. Therefore, it may happen that $f$ is $k$-linear and the projection gives a $(k-1)$-linear$^*$ function.

\subsection{The Lower Bound for One-sided Testers}
We first show the result for non-adaptive testers.
Suppose there is a one-sided non-adaptive uniform distribution $1/8$-tester $A(s,f)$ for $k$-Linear that makes $q$ queries, where $s$ is the random seed of the tester and $f$ is the function that is tested. The algorithm has access to $f$ through a black box queries.

Consider the set of linear functions $C=\{g^{(0)}\}\cup\{g^{(\ell)}=x_n+\cdots+x_{n-\ell+1}|\ell=1,\ldots,k-1\}\subseteq (k-1)$-Linear$^*$ where $g^{(0)}=0$. Any $k$-linear function is $1/2$-far from every function in $C$ w.r.t. the uniform distribution. Therefore, using the tester $A$, with probability at least $2/3$, we can distinguish between any $k$-linear and any function in $C$. By running the tester $A$ $O(\log k)$ times, and accept if and only if all accept, we get a tester $A'$ that asks $O(q\log k)$ queries and satisfies
\begin{enumerate}
\item If $f\in k$-Linear then with probability $1$, $A'(s,f)$ accepts.
\item If $f\in C$ then, with probability at least $1-1/(2k)$, $A'(s,f)$ rejects.
\end{enumerate}
By an averaging argument (i.e., fixing coins for $A'$) and since $|C|=k$, there exists a deterministic non-adaptive algorithm $B$ that makes $q'=O(q\log k)$ queries such that
\begin{enumerate}
\item If $f\in k$-Linear then $B(f)$ accepts.
\item If $f=C$ then $B(f)$ rejects.
\end{enumerate}

Let $a^{(i)}$, $i=1,\ldots,q'$ be the queries that $B$ makes. Let $M$ be a $q'\times n$ binary matrix where the $i$-th row of $M$ is $a^{(i)}$ and $x^f\in \{0,1\}^n$ where $x^f_i=1$ if $i$ is a relevant coordinate in $f$. Then the vector of answers to the queries of $B(f)$ is $Mx^f$. If $Mx^f=Mx^g$ for some $g\in C$, that is, the answers of the queries to $f$ are the same as the answer of the queries to $g$, then $B(f)$ rejects. Therefore, for every $f\in k$-Linear and every $g\in C$ we have $Mx^f\not=Mx^g$. Now since $\{x^f|f\in k{\rm -Linear}\}$ is the set of all strings of weight $k$, the sum (over the field $F_2$) of every $k$ columns of $M$ is not equal to $0$ and not equal to the sum of the last $\ell$ columns of $M$, for all $\ell=1,\ldots,k-1$. In particular, if $M_i$ is the $i$th column of $M$, for every $i_1,\ldots,i_{k-\ell}\le n-k+1$,
$M_{i_1}+\cdots+M_{i_{k-\ell}}+M_{n-\ell+1}+\cdots+M_n\not= M_{n-\ell+1}+\cdots+M_n$ and therefore $M_{i_1}+\cdots+M_{i_{k-\ell}}\not=0$. That is, the sum of every less or equal $k-1$ columns of the first $n-k+1$ columns of $M$ is not equal to zero. We then show that such matrix has at least $q'=\Omega(k\log n)$ rows. This implies that $q=\Omega((k/\log k)\log n)$.

For the lower bound for adaptive testers we take $C=\{g^{(\ell)}\}$ for some $\ell\in \{0,1,\ldots,k-1\}$ and get a $q\times n$ matrix $M$ that the sum of every $k-\ell$ columns of $M$ is not zero. We then show that there exists $\ell\le k-1$ where such a matrix must have at least $q=\tilde \Omega(\sqrt{k}\log n)$ rows.

\section{The Testers for $k$-Linear$^*$ and $k$-Linear}
In this section we give the non-adaptive distribution-free one-sided tester for $k$-Linear$^*$ and the non-adaptive distribution-free two-sided tester for $k$-Linear.

\subsection{Notations}
In this subsection, we give some notations that we use throughout the paper.

Denote $[n]=\{1,2,\ldots,n\}$. For $S\subseteq [n]$ and $x=(x_1,\ldots,x_n)$. For $X\subset [n]$ we denote by $\{0,1\}^X$
the set of all binary strings of
length $|X|$ with coordinates indexed by $i\in X$. For $x\in \{0,1\}^n$ and $X\subseteq [n]$ we write $x_X\in\{0,1\}^{X}$ to denote the projection of $x$ over coordinates in $X$. We denote by $1_X$ and $0_X$ the all-one and all-zero strings in $\{0,1\}^{X}$, respectively. For a variable $x_i$ and a set $X$, we denote by $(x_i)_X$ the string $x'$ over coordinates in $X$ where for every $j\in X$, $x'_j=x_i$.
For $X_1,X_2\subseteq [n]$ where $X_1\cap X_2=\emptyset$ and $x\in \{0,1\}^{X_1}, y\in \{0,1\}^{X_2}$ we write $x\circ y$  to denote their concatenation, i.e.,
the string in $\{0,1\}^{X_1\cup X_2}$ that agrees with $x$ over coordinates in $X_1$ and agrees with $y$ over coordinates in~$X_2$. For $X\subseteq [n]$ we denote $\overline{X}=[n]\backslash X=\{x\in [n]|x\not\in X\}$.

For example, if $n=7$, $X_1=\{1,3,5\}$, $X_2=\{2,7\}$, $y_2$ is a variable and $z=(z_1,z_2,z_3,z_4,z_5,z_6,z_7)$ $\in\{0,1\}^7$ then $(y_2)_{X_1}\circ z_{X_2}\circ 0_{\overline{X_1\cup X_2}}=(y_2,z_2,y_2,0,y_2,0,z_7)$.

\subsection{The Tester}
Consider the tester {\bf Test-Linear$^*_k$} for $k$-Linear$^*$ in Figure~\ref{ALg01}. The tester uses three procedures. The first is {\bf Self-corrector} that for an input $x\in\{0,1\}^n$ chooses a uniform random $z\in\{0,1\}^n$ and returns $f(x+z)+f(z)$. The procedure {\bf BLR} that is a non-adaptive uniform-distribution one-sided $\epsilon$-tester for Linear. BLR makes $c_1/\epsilon$ queries for some constant $c_1$,~\cite{BlumLR93}. The third procedure is {\bf Hoffmeister's Algorithm} $(N,K)$, a deterministic non-adaptive algorithm that exactly learns $K$-Linear$^*$ over $N$ coordinates from black box queries. Hoffmeister's Algorithm makes $c_2K\log N$ queries for some constant $c_2$,~\cite{Hofmeister99}.
\newcounter{ALg}
\setcounter{ALg}{0}
\newcommand{\stepal}{\stepcounter{ALg}$\arabic{ALg}.\ $\>}
\newcommand{\steplabelal}[1]{\addtocounter{ALg}{-1}\refstepcounter{ALg}\label{#1}}
\begin{figure}[h!]
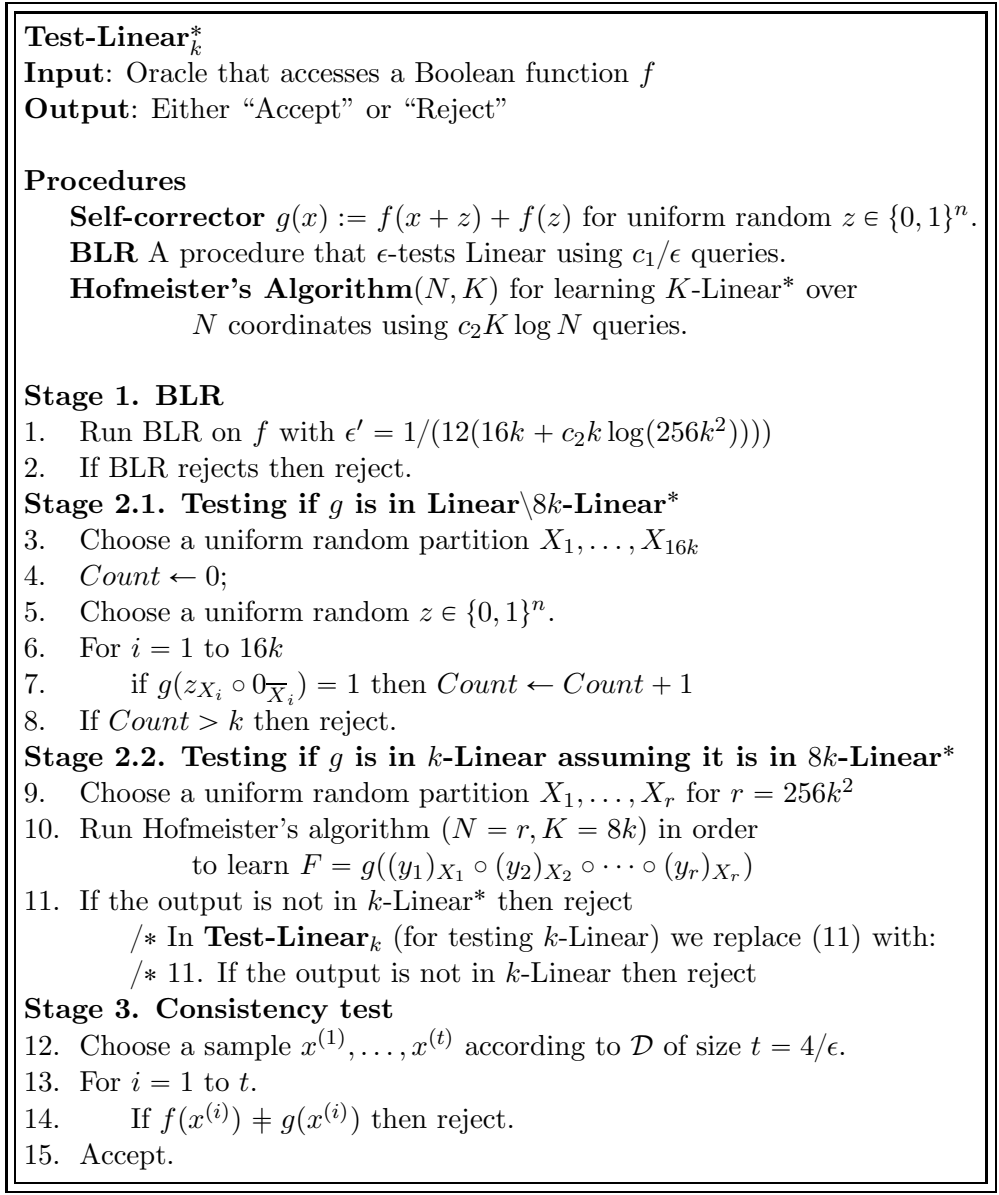

  \begin{center}
  \fbox{\fbox{\begin{minipage}{28em}
  \begin{tabbing}
  xxx\=xxxx\=xxxx\=xxxx\=xxxx\=xxxx\= \kill
  {{\bf Test-Linear$_k^*$}}\\
{\bf Input}: Oracle that accesses a Boolean function $f$  \\

  {\bf Output}: Either ``Accept'' or ``Reject''\\ \\
    {\bf Procedures}\\
\> {\bf Self-corrector} $g(x):=f(x+z)+f(z)$ for uniform random $z\in\{0,1\}^n$.\\
\> {\bf BLR} A procedure that $\epsilon$-tests Linear using $c_1/\epsilon$ queries.\\
\> {\bf Hofmeister's Algorithm}$(N,K)$ for learning $K$-Linear$^*$ over\\
\>\>\>  $N$ coordinates using $c_2K\log N$ queries.\\ \\

{\bf Stage 1. BLR}\\
\stepal\steplabelal{ALg09}
Run BLR on $f$ with $\epsilon'=1/(12(16k+c_2k\log (256k^2))))$ \\
\stepal\steplabelal{ALg10}
If BLR rejects then reject.\\
{\bf Stage 2.1. Testing if $g$ is in Linear$\backslash 8k$-Linear$^*$}\\
\stepal\steplabelal{ALg20}
Choose a uniform random partition $X_1,\ldots,X_{16k}$\\
\stepal\steplabelal{ALg21}
$Count\gets 0$; \\
\stepal\steplabelal{ALg26}
Choose a uniform random $z\in \{0,1\}^n$.\\
\stepal\steplabelal{ALg22}
For $i=1$ to $16k$\\
\stepal\steplabelal{ALg23}
\>   if $g(z_{X_i}\circ 0_{\overline X_i})=1$ then $Count\gets Count +1$ \\
\stepal\steplabelal{ALg24}
If $Count>k$ then reject. \\
{\bf Stage 2.2. Testing if $g$ is in $k$-Linear assuming it is in $8k$-Linear$^*$}\\
\stepal\steplabelal{ALg30}
Choose a uniform random partition $X_1,\ldots,X_{r}$ for $r=256k^2$\\
\stepal\steplabelal{ALg40}
Run Hofmeister's algorithm $(N=r,K=8k)$ in order \\
\>\>\> to learn $F=g((y_1)_{X_1}\circ (y_2)_{X_2}\circ \cdots \circ(y_r)_{X_r})$\\
\stepal\steplabelal{ALg50}
If the output is not in $k$-Linear$^*$ then reject \\
\>\> $/*$ In {\bf Test-Linear$_k$} (for testing $k$-Linear) we replace (11) with: \\
\>\> $/*$ \ref{ALg50}. If the output is not in $k$-Linear then reject\\
{\bf Stage 3. Consistency test}\\
\stepal\steplabelal{ALg65}
Choose a sample $x^{(1)},\ldots, x^{(t)}$ according to $\D$ of size $t=4/\epsilon$.\\
\stepal\steplabelal{ALg70}
For $i=1$ to $t$.\\
\stepal\steplabelal{ALg80}
\> If $f(x^{(i)})\not= g(x^{(i)})$ then reject.\\
\stepal\steplabelal{ALg81}
Accept.
  \end{tabbing}
  \end{minipage}}}
  \end{center}
	\caption{An optimal two-sided tester for $k$-Linear.}
	\label{ALg01}
\end{figure}

To test $k$-Linear we use the same tester but change step \ref{ALg50} to:

\noindent
(\ref{ALg50}) If the output is not in $k$-Linear then reject

We call this tester {\bf Test-Linear$_k$}.

\subsection{Correctness of the Tester}
In this section we prove
\begin{theorem} {\bf Test-Linear$_k$} is a non-adaptive distribution-free two-sided $\epsilon$-tester for $k$-Linear that makes
$$O\left(k\log k +\frac{1}{\epsilon}\right)$$ queries.
\end{theorem}

\begin{theorem} {\bf Test-Linear$^*_k$} is a non-adaptive distribution-free one-sided $\epsilon$-tester for $k$-Linear$^*$ that makes
$$O\left(k\log k +\frac{1}{\epsilon}\right)$$ queries.
\end{theorem}

\begin{proof} Since there is no stage in the tester that uses the answers of the queries asked in previous ones, the tester is non-adaptive.

In Stage 1 the tester makes $O(1/\epsilon')=O(k\log k)$ queries. In stage~2.1, $O(k)$ queries. In stage~2.2, $O(k\log r)=O(k\log k)$ queries and in stage~3, $O(1/\epsilon)$ queries. Therefore, the query complexity of the tester is $O(k\log k+1/\epsilon)$.

We will assume that $k\ge 12$. For $k<12$, (see the introduction and Table~\ref{Table}) the non-adaptive tester of $k$-Junta with the BLR tester and the self-corrector gives a non-adaptive testers that makes $O(1/\epsilon)=O(k\log k+1/\epsilon)$ queries.

\vspace{10px}
\noindent
{\bf Completeness}: We first show the completeness for {\bf Test-Linear$_k$} that tests $k$-Linear. Suppose $f\in k$-Linear. Then for every $x$ we have $g(x)=f(x+z)+f(z)=f(x)+f(z)+f(z)=f(x)$. Therefore, $g=f$. In stage 1, BLR is one-sided and therefore it does not reject. In stage 2.1, since $X_1,\ldots,X_{16k}$ are pairwise disjoint, the number of functions $g(x_{X_i}\circ 0_{\overline{X_i}})$, $i=1,2,\ldots, 16k$, that are not identically zero is at most $k$ and therefore stage 2.1 does not reject. In stage 2.2, with probability at least $1-{k\choose 2}/(256 k^2)\ge2/3$, the relevant coordinates of $f$ fall into different $X_i$ and then $F=g((y_1)_{X_1}\circ (y_2)_{X_2}\circ \cdots \circ(y_r)_{X_r})=f((y_1)_{X_1}\circ (y_2)_{X_2}\circ \cdots \circ(y_r)_{X_r})$ is $k$-linear. Then, Hofmeister's algorithm returns a $k$-linear function. Therefore, with probability at least $2/3$ the tester does not reject. Stage 3 does not reject since $f=g$.

Now for the tester {\bf Test-Linear$^*_k$}, in stage 2.2, with probability $1$ the function $F$ is in $k\sLinear^*$. In fact, if $t$ relevant coordinates falls into the set $X_i$ then the coordinate $i$ (that correspond to the variable $y_i$) will be relevant in $F$ if and only if $t$ is odd. Therefore, the tester does not reject.

Notice that {\bf Test-Linear$^*_k$} is one-sided and {\bf Test-Linear$_k$} is two-sided.

\vspace{10px}
\noindent
{\bf Soundness}: We prove the soundness for {\bf Test-Linear$_k$}. The same proof also works for {\bf Test-Linear$^*_k$}.
Suppose $f$ is $\epsilon$-far from $k$-Linear w.r.t. the distribution $\D$. We have four cases
\begin{description}
\item[Case 1]: $f$ is $\epsilon'$-far from Linear w.r.t. the uniform distribution.
\item[Case 2]: $f$ is $\epsilon'$-close to $g\in$Linear and $g$ is in $\Linear\backslash 8k\sLinear^*$.
\item[Case 3]: $f$ is $\epsilon'$-close to $g\in$Linear and $g$ is in $8k\sLinear^*\backslash k\sLinear$.
\item[Case 4]: $f$ is $\epsilon'$-close to $g\in$Linear, $g$ is in $k\sLinear$ and $f$ is $\epsilon$-far from $k\sLinear$ w.r.t. $\D$.
\end{description}

For Case 1, if $f$ is $\epsilon'$-far from Linear then, in stage 1, BLR rejects with probability $2/3$.

For Cases 2 and 3, since $f$ is $\epsilon'$-close to $g$, for any fixed $x\in\{0,1\}^n$ with probability at least $1-2\epsilon'$ (over a uniform random $z$), $f(x+z)+f(z)=g(x+z)+g(z)=g(x)$. Since stages~2.1 and~2.2 makes $(16k+c_2k\log r)$ queries (to $g$), with probability at least $1-(16k+c_2k\log r)2\epsilon'\ge 5/6$, $g(x)$ is computed correctly for all the queries in stages 2.1 and 2.2.

For Case 2, consider stage 2.1 of the tester. If $g$ is in $\Linear\backslash 8k\sLinear^*$ then $g$ has more than $8k$ relevant coordinates. The probability that less than or equal to $4k$ of $X_1,\ldots,X_{16k}$ contains relevant coordinates of $g$ is at most
$${16k\choose 4k}\frac{1}{4^{8k}}\le \left(\frac{e16k}{4k}\right)^{4k}\frac{1}{4^{8k}}\le \frac{1}{12}.$$

If $X_i$ contains the relevant coordinates $i_1,\ldots,i_\ell$ then $g(x_{X_i}\circ 0_{\overline X_i})=x_{i_1}+\cdots+x_{i_\ell}$ and therefore, for a uniform random $z\in\{0,1\}^n$, with probability at least $1/2$, $g(z_{X_i}\circ 0_{\overline X_i})=1$. Therefore, if at least $4k$ of $X_1,\ldots,X_{16k}$ contains relevant coordinates then, by Chernoff bound, with probability at least $1-e^{-k/4}\ge 11/12$, the counter ``{\it Count}'' is greater than $k$. Therefore, for Case 2, if $g$ is in $\Linear\backslash 8k\sLinear^*$ then, with probability at least $1-(1/6+1/12+1/12)=2/3$, the tester rejects.

For Case 3, consider stage 2.2. If $g$ is in $8k\sLinear^*\backslash k\sLinear$ then $g$ has at most $8k$ relevant coordinates. Then with probability at least $1-{8k\choose 2}/(256 k^2)\ge 5/6$, the relevant coordinates of $g$ fall into different $X_i$ and then Hofmeister's algorithm returns a linear function with the same number of relevant coordinates as $g$. Therefore stage 2.2 rejects with probability at least $2/3$.

For Case 4, if $g$ is in $k$-Linear and $f$ is $\epsilon$-far from $k$-Linear w.r.t. $\D$, then $f$ is $\epsilon$-far from $g$ w.r.t.~$\D$. Then for uniform random $z$ and $x\sim \D$,
\begin{eqnarray*}
\Pr_{\D,z}[f(x)\not=g(x)]&\ge& \Pr_{\D,z}[f(x)\not=g(x)|g(x)=f(x+z)+f(z)]\Pr_{\D,z}[g(x)=f(x+z)+f(z)]\\
&=& \Pr_{\D}[f(x)\not=g(x)]\Pr_{z}[g(x)=f(x+z)+f(z)]\\
&\ge& \epsilon (1-\epsilon')\ge \epsilon/2.
\end{eqnarray*}
Therefore, with probability at most $(1-\epsilon/2)^{t}=(1-\epsilon/2)^{4/\epsilon}\le 1/3$, stage 3 does not reject.
\end{proof}

\section{Lower Bound}
In this section we prove
\begin{theorem} Any non-adaptive uniform-distribution one-sided $1/8$-tester for $k$-Linear must make at least  $\tilde\Omega(k\log n)$ queries.
\end{theorem}

\begin{theorem} Any adaptive uniform-distribution one-sided $1/8$-tester for $k$-Linear must make at least  $\tilde\Omega(\sqrt{k}\log n)$ queries.
\end{theorem}

\subsection{Lower Bound for Non-Adaptive Testers}

We first show the result for non-adaptive testers.

Suppose there is a non-adaptive uniform-distribution one-sided $1/8$-tester $A(s,f)$ for $k$-Linear that makes $q$ queries, where $s$ is the random seed of the tester and $f$ is the function that is tested. The algorithm has access to $f$ through a black box queries.

Consider the set of linear functions $C=\{g^{(0)}\}\cup\{g^{(\ell)}=x_n+\cdots+x_{n-\ell+1}|\ell=1,\ldots,k-1\}\subseteq (k-1)$-Linear$^*$ where $g^{(0)}=0$. Any $k$-linear function is $1/2$-far from every function in $C$ w.r.t. the uniform distribution. Therefore, using the tester $A$, with probability at least $2/3$, $A$ can distinguish between any $k$-linear function and functions in $C$.  We boost the success probability to $1-1/(2k)$ by running $A$, $\log (2k)/\log 3$ times, and accept if and only if all accept. We get a tester $A'$ that asks $O(q\log k)$ queries and satisfies
\begin{enumerate}
\item If $f\in k$-Linear then with probability $1$, $A'(s,f)$ accepts.
\item If $f\in C$ then, with probability at least $1-1/(2k)$, $A'(s,f)$ rejects.
\end{enumerate}

Therefore, the probability that for a uniform random $s$, $A'(s,f)$ accepts for some $f\in C$ is at most $1/2$. Thus, there is a seed $s_0$ such that $A'(s_0,f)$ rejects for all $f\in C$ (and accept for all $f\in k$-Linear).
This implies that there exists a deterministic non-adaptive algorithm $B(=A'(s_0,*))$ that makes $q'=O(q\log k)$ queries such that
\begin{enumerate}
\item If $f\in k$-Linear then $B(f)$ accepts.
\item If $f\in C$ then $B(f)$ rejects.
\end{enumerate}

Let $a^{(i)}$, $i=1,\ldots,q'$ be the queries that $B$ makes. Let $M$ be a $q'\times n$ binary matrix that it's $i$-th row is $a^{(i)}$. Let $x^f\in \{0,1\}^n$ where $x^f_i=1$ iff $i$ is relevant coordinate in $f$. Then the vector of answers to the queries of $B(f)$ is $Mx^f$. If $Mx^f=Mx^g$ for some $g\in C$, that is, the answers of the queries to $f$ are the same as the answers of the queries to $g$, then $B(f)$ rejects. Therefore, for every $f\in k$-Linear and every $g\in C$ we have $Mx^f\not=Mx^g$. Now since $\{x^f|f\in k{\rm -Linear}\}$ is the set of all strings of weight $k$, the sum (over the field $F_2$) of every $k$ columns of $M$ is not equal to $0$ (zero string) and not equal to the sum of the last $\ell$ columns of $M$, for all $\ell=1,\ldots,k-1$. In particular, if $M_i$ is the $i$th column of $M$, for every $i_1,\ldots,i_{k-\ell}\le n-k+1$,
$M_{i_1}+\cdots+M_{i_{k-\ell}}+M_{n-\ell+1}+\cdots+M_n\not= M_{n-\ell+1}+\cdots+M_n$ and therefore $M_{i_1}+\cdots+M_{i_{k-\ell}}\not=0$. That is, the sum of every less or equal $k$ columns of the first $n-k+1$ columns of $M$ is not equal to zero. We then show in Lemma~\ref{Hamming} that such matrix has at least $q'=\Omega(k\log n)$ rows. This implies that $q=\Omega((k/\log k)\log n)$.

Let $\pi(n,k)$ be the minimum integer $q$ such that there exists a $q\times n$ matrix over $F_2$ that the sum of any of its less than or equal $k$ columns is not $0$. We have proved

\begin{lemma}\label{Lower1} Any non-adaptive uniform-distribution one-sided $1/8$-tester for $k$-Linear must make at least  $\Omega(\pi(n-k+1,k)/\log k)$ queries.
\end{lemma}

Now to show that $\Omega(\pi(n-k+1,k)/\log k)=\Omega(k\log n)$ we prove the following result. This lemma follows from Hamming's bound in coding theory. We give the proof for completeness
\begin{lemma}\label{Hamming} (Hamming's Bound) We have
$$\pi(n,k)\ge \log \sum_{i=0}^{\left\lfloor \frac{k}{2}\right\rfloor}{n\choose i}=\Omega(k\log (n/k)).$$
\end{lemma}
\begin{proof}\label{lower11} Let $M$ be a $\pi(n,k)\times n$ matrix over $F_2$ that the sum of any of its less than or equal $k$ columns is not $0$. Let $m=\lfloor k/2\rfloor$ and $S=\{M_{i_1}+\cdots+M_{i_t}\ |\ t\le m \mbox{\ and\ } 1\le i_1<\cdots<i_t\le n \}\subseteq\{0,1\}^{\pi(n,k)}$ be a multiset. The strings in $S$ are distinct because, if for the contrary, we have two strings in $S$ that satisfies $M_{i_1}+\cdots+M_{i_t}=M_{j_1}+\cdots+M_{j_{t'}}$ then $M_{i_1}+\cdots+M_{i_t}+M_{j_1}+\cdots+M_{j_{t'}}=0$ (equal columns are cancelled) and $t+t'\le k$, which is a contradiction. Therefore, $2^{\pi(n,k)}\ge |S|=\sum_{i=0}^{m}{n\choose i}$ and $\pi(n,k)\ge \log |S|$.
\end{proof}

\subsection{Lower Bound for Adaptive Testers}

For the lower bound for adaptive testers we take $C=\{g^{(\ell)}\}$ for some $\ell\in \{0,1,\ldots,k-1\}$ and get an adaptive algorithm $A$ that makes $q$ queries and satisfies
\begin{enumerate}
\item If $f\in k$-Linear then with probability $1$, $A(s,f)$ accepts.
\item If $f=g^{(\ell)}$ then, with probability at least $2/3$, $A(s,f)$ rejects.
\end{enumerate}
This implies that there exists a deterministic adaptive algorithm $B=A(s_0,*)$ that makes $q$ queries such that
\begin{enumerate}
\item If $f\in k$-Linear then $B(f)$ accepts.
\item If $f=g^{(\ell)}$ then $B(f)$ rejects.
\end{enumerate}
Then, by the same argument as in the case of non-adaptive tester, we get a $q\times n$ matrix $M$ that the sum of every $k-\ell$ columns of the first $n-\ell$ columns of $M$ is not zero. Let $\Pi(n,k)$ be the minimum integer~$q$ such that there exists a $q\times n$ matrix over $F_2$ that the sum of any of its $k$ columns is not~$0$. Then, we have proved that

\begin{lemma}\label{Lower2} Any adaptive uniform-distribution one-sided $1/8$-tester for $k$-Linear must make at least  $\Omega(\max_{1\le \ell\le k}\Pi(n-k,\ell))$ queries.
\end{lemma}

In the next subsection, we show that there exists $1\le\ell\le k$ such that $\Pi(n,\ell)=\tilde\Omega(\sqrt{k}\log n)$.

\subsection{A Lower Bound for $\Pi$}

In this section we prove
\begin{lemma} We have $\max_{1\le \ell\le k}\Pi(n,\ell)=\tilde\Omega(\sqrt{k}\log n)$.
\end{lemma}

The idea of the proof is the following. For a set of integers $L$ an $L$-{\it good matrix} $M$ is a matrix that for every $\ell\in L$ the sum of every $\ell$ columns of $M$ is not zero. A $k$-good matrix is a $\{k\}$-good matrix. We say that the matrix $M$ is {\it almost $L$-good} if there is a ``small'' number ($poly(k)$) of columns of $M$ that can be removed to get an $L$-good matrix. The concatenation $M_1\circ M_2$ (the matrix that contains the rows of both matrices) of almost $L_1$-good matrix $M_1$ with an almost $L_2$-good matrix $M_2$ is an almost $L_1\cup L_2$-good matrix.

Let $K=\lfloor \sqrt{k}/(2\log k)\rfloor$ and $[K]=\{1,2,\ldots,K\}$. The idea of the proof is to construct an almost $[K]$-good matrix $M$ by concatenating $t=O(\log k)$ matrices $M_1\circ M_2\circ\cdots\circ M_t$ where $M_i$ is $k_i$-good $(\Pi(n,k_i)\times n)$-matrices for some $k_i\le k$. Then after removing small number ($poly(k)$) columns of $M$ we get a $[K]$-good matrix $M$ with $\sum_{i=1}^t \Pi(n,k_i)$ rows and $n-poly(k)$ columns. By Hamming's bound, Lemma~\ref{Hamming},  $M$ contains at least $\Omega(K\log n)$ rows. Therefore, $\sum_{i=1}^t \Pi(n,k_i)= \Omega(K\log n)$. So there is $i$ such that $\Pi(n,k_i)=\Omega(K\log n/\log k)=\Omega(\sqrt{k}\log n/\log^2 k)=\tilde \Omega(\sqrt{k}\log n)$.

We now give more intuition to how to construct an almost $[K]$-good matrix from $k_i$-good matrices. Denote by $\N_d=\{i: d\notdivides i\}\cap [K]$. Let $k=k_1$. We first show that if $M_1$ is $k_1$-good matrix then there exists a set of integers $L_1\subseteq [K]$ such that $M_1$ is almost $L_1$-good matrix and $d_1:=\gcd([K]\backslash L_1)\notdivides k_1$. The intuition is that if, for the contrary, there are many pairwise disjoint sets of columns that sum to $0$ that the great common divisor of their sizes divides $k_1$, then the union of some of them gives $k_1$-set of columns that sum to $0$ and then we get a contradiction.
Therefore $d_1\not=1$, $L_1\supseteq \N_{d_1}$ and $M_1$ is almost $\N_{d_1}$-good. We then take $k_2:=d_1\lfloor k/d_1\rfloor$ and a $k_2$-good $\Pi(n,k_2)\times n$ matrix $M_2$. Then, as before, $M_2$ is almost $\N_{d_2}$-good matrix
with $d_2\notdivides k_2$. Therefore, $d_2\notdivides d_1$. Now the concatenation of both matrices $M_1\circ M_2$ is almost $\N_{d_1}\cup \N_{d_2}=\N_{\lcm(d_1,d_2)}$. Since $d_2\notdivides d_1$ we must have $d_2':=\lcm(d_1,d_2)\ge 2d_1$. We then take $k_3=d_2'\lfloor k/d_2'\rfloor$ and a $k_3$-good $\Pi(n,k_3)\times n$ matrix $M_3$ and concatenate it with $M_1\circ M_2$ to get an almost $\N_{\lcm(d_1,d_2,d_3)}$-good matrix with $\lcm(d_1,d_2,d_3)\ge 2d_2'=2\lcm(d_1,d_2)\ge 4d_1$. After, $t=O(\log k)$ iterations, we get a $(\sum_{i=1}^t \Pi(n,k_i))\times n$ matrix $M_1\circ M_2\circ \cdots \circ M_t$ that is almost $\N_d$-good for some $d\ge 2^td_1>K$ and therefore, almost $[K]$-good.

We note here that we can get the bound $\Omega(\sqrt{k}(\log\log k)\log n/\log^2 k)$ by choosing $k_1=\lcm(1,2,3,$ $\cdots,m_i)\le k$, and then $k_i=d_{i-1}' \lcm(1,2,3,\cdots,m_i)<k$ where $m_i=O(\log(k))$. See \cite{Farhi09}.

We now give the full proof.
We start with some preliminary results, Lemmas~\ref{N1}-\ref{zero}.
\begin{lemma}\label{N1} Let $W\subseteq [m]$ and $w=\gcd(W)$. There is a subset $W'\subseteq W$ of size
$$O\left(\frac{\log \frac{m}{w}}{\log\log  \frac{m}{w}}\right)<\log \frac{m}{w}$$ such that $\gcd(W')=\gcd(W)$.
\end{lemma}
\begin{proof} Define the set $D=W/w=\{b/w|b\in W\}$. Then $D\subseteq [\lfloor m/w\rfloor]$ and $\gcd(D)=1$. Let $D'\subseteq D$ be a minimum size set with $\gcd(D')=1$ and $W'=wD'\subseteq W$. Let $D'=\{d_1,\ldots,d_t\}$ and $g_i=\gcd(D'\backslash \{d_i\})$ for $i=1,\ldots,t$. Since $D'$ is minimum $g_i>1$. We also have for $i\not=j$, $$1=\gcd(D')=\gcd(\gcd(D'\backslash \{d_i\}),\gcd(D'\backslash \{d_j\}))=\gcd(g_i,g_j)$$ and therefore $g_1,\ldots,g_t$ are pairwise relatively prime. Since for all $i>1$, $g_i=\gcd(D'\backslash \{d_i\})|d_1$ we have $\prod_{i=2}^tg_i| d_1$. Therefore,
$\lfloor m/w\rfloor \ge d_1\ge \prod_{i=2}^tg_i\ge \prod_{i=2}^ti=t!=|D'|!=|W'|!$ and the result follows.
\end{proof}

\begin{lemma}\label{tin} Let $d,d',k,y\ge 1$ be integers that satisfy $d|y, d|k, d|d'$ and $\gcd(y,d')=d$. There is $0\le \lambda<d'/d$ such that $d'|(k-\lambda y)$.
\end{lemma}
\begin{proof} Let $\hat y=y/d,\hat k=k/d$ and $\hat d=d'/d$. Then $\gcd(\hat y,\hat d)=1$. Consider the set $B=\{\hat k-i\hat y \ |\ i=0,\ldots,\hat d-1\}$. If for $0\le i_1<i_2\le \hat d-1$ we have $\hat k-i_1\hat y  = (\hat k-i_2\hat y\mod \hat d)$ then $(i_1-i_2)\hat y=(0\mod \hat d)$. Since $\gcd(\hat y,\hat d)=1$ we get $i_1=(i_2 \mod \hat d)$ and therefore $i_1=i_2$. This shows that the elements in $B$ are distinct modulo $\hat d$ and therefore there is $0\le \lambda<\hat d=d'/d$ such that $\hat k-\lambda \hat y =(0\mod \hat d)$. Then $ k-\lambda  y =(0\mod  d')$.
\end{proof}

\begin{lemma} \label{Lamb} Let $k$ be an integer. Let $J=\{j_1,\ldots,j_\ell\}$ be a set of integers such that $1\le j_1,\ldots,j_\ell\le \sqrt{k}/\ell$ and $d:=\gcd(j_1,\ldots,j_\ell)|k$. There exist non-negative integers $0\le \lambda_1,\ldots,\lambda_{\ell-1}\le \sqrt{k}$ and $0\le \lambda_\ell\le k$ such that
$$\lambda_1j_1+\lambda_2j_2+\cdots+\lambda_\ell j_\ell=k.$$
\end{lemma}
\begin{proof} We prove the result by induction on $\ell$. For $\ell=1$, given $J=\{j_1\}$, $1\le j_1\le \sqrt{k}$ and $d=j_1|k$ we let $\lambda_1=k/d$. Then $\lambda_1\le k$ and $\lambda_1 j_1=k$.

Assume that the result is true for $\ell-1$. We prove the result for $\ell$.

Given $d:=\gcd(j_1,\ldots,j_\ell)|k$. Let $d'=\gcd(j_2,\ldots,j_{\ell})$. We have two cases: $d'=d$ and $d'>d$.
If $d'=d$ then $d'|k$ and for $i>1$, $j_i\le \sqrt{k}/\ell\le \sqrt{k}/(\ell-1)$. By the induction hypothesis there are $0\le \lambda_2,\ldots,\lambda_{\ell-1}\le \sqrt{k}$ and $0\le \lambda_\ell\le k$ such that $\lambda_2j_2+\lambda_3j_3+\cdots+\lambda_\ell j_\ell=k.$ We choose $\lambda_1=0$ and the result follows.

Now suppose $d'>d$. We have $d|j_1$, $d|k$, $d|d'$ and $\gcd(j_1,d')=d$. By Lemma~\ref{tin}, there is $\lambda_1$ such that $0\le \lambda_1<d'/d$ and $d'|k':=k-\lambda_1j_1$. Since $\lambda_1<d'/d\le j_2\le \sqrt{k}$, we also have
$$k'=k-\lambda_1j_1 \ge k-\sqrt{k}\sqrt{k}/\ell=\frac{\ell-1}{\ell}k$$ and therefore $j_2,\ldots,j_\ell \le \sqrt{k}/\ell\le \sqrt{k'}/(\ell-1)$. Since $d'|k'$, by the induction hypothesis there exist $0\le \lambda_2,\ldots,\lambda_{\ell-1}\le \sqrt{k'}\le \sqrt{k}$ and $0\le \lambda_\ell\le k'<k$ such that $\lambda_2j_2+\lambda_3j_3+\cdots+\lambda_\ell j_\ell=k'.$ Then $\lambda_1j_1+\lambda_2j_2+\cdots+\lambda_\ell j_\ell=k.$
\end{proof}

Let $M$ be a $q\times n$ binary matrix. Recall that $M_i$ is the $i$th column of $M$. For every $j\ge 1$, let $\ell_j(M)$ denotes the maximum number of disjoint $j$-subsets $A_1,A_2,\ldots$ of $[n]$ such that $\sum_{j\in A_i}M_j=0$ for all~$i$. We say that $M$ is $(j,\ell)$-{\it good} if $\ell_j(M)\le \ell$ and $(j,\ell)$-{\it bad} if it is not $(j,\ell)$-good, i.e., $\ell_j(M)> \ell$. For $L,J\subseteq [n]$, we say that $M$ is $(L,\ell)$-good if it is $(j,\ell)$-good for all $j\in L$ and $(J,\ell)$-bad if it is $(g,\ell)$-bad for all $j\in J$. When $\ell=0$ we just say $j$-good, $L$-good, $j$-bad and $J$-bad.

For two $q_1\times n$ and $q_2\times n$ matrices $M$ and $M'$, respectively, the {\it concatenation of $M$ and $M'$} is $M\circ M'=[M^*|M'^*]^*$ where $*$ is the transpose of a matrix. That is, $M\circ M'$ is the $(q_1+q_2)\times n$ matrix that results from the rows of $M$ follows by the rows of $M'$.

The following result is obvious
\begin{lemma}\label{concatenation} If $M$ is $(L,\ell)$-good and $M'$ is $(L',\ell)$-good then $M \circ M'$ is $(L\cup L',\ell)$-good.
\end{lemma}

\begin{lemma}\label{llbb} Let $M$ be a $q\times n$ matrix. If $M$ is $([d],\ell)$-good then $q= \Omega(d\log ((n-(\ell d^2/2))/d))$.
\end{lemma}
\begin{proof}
For every $j\in [d]$ we have $\ell_j(M)\le \ell$. That is, for every $j$, there are at most $\ell$ disjoint $j$-sets of columns that sum to zero. We remove those columns (for all $j\in[d]$) and get a $([d],0)$-good matrix. The number of columns that are removed is at most $\sum_{j=1}^d\ell j\le \ell d^2/2$. Using Hamming's bound, Lemma~\ref{Hamming}, the result follows.
\end{proof}

We now prove
\begin{lemma}\label{zero} Let $m$, $q$, $w$ and $t=mqw$ be integers. Let $J=\{j_1,\ldots,j_w\}\subseteq [m]$. Let $M$ be a $(J,t)$-bad matrix.
Then for any $\lambda_1,\ldots,\lambda_w\in [q]$ we have that $M$ is $(\lambda_1j_1+\cdots+\lambda_wj_w)$-bad.
\end{lemma}
\begin{proof} Let $r=\lambda_1j_1+\cdots+\lambda_wj_w$. We need to show that there are $r$ columns of $M$ that sum to $0$.
Since $M$ is $(j_1,t)$-bad and $\lambda_1\le t$, there are $\lambda_1$ pairwise disjoint $j_1$-sets $A_{1,1},A_{1,2},\cdots, A_{1,\lambda_1}$ such that $\sum_{j\in A_{1,i}}M_j=0$ for all $i\in [\lambda_1]$. Since $M$ is $(j_2,t)$-bad and $\lambda_2\le t-\lambda_1j_1$, there are $\lambda_2$ pairwise disjoint $j_2$-sets $A_{2,1},A_{2,2},\cdots, A_{2,\lambda_2}$ sets that are also pairwise disjoint with $A_{1,1},A_{1,2},\cdots, A_{1,\lambda_1}$ such that $\sum_{j\in A_{2,i}}M_j=0$ for all $i\in [\lambda_2]$. We continue with this procedure until we find a collection ${\cal A}'$ of disjoint sets that contains, for every $i\le w-1$, $\lambda_i$ $j_i$-sets that corresponds to columns of $M$ that sum to $0$. Now since $\lambda_w\le t-(\lambda_1j_1+\cdots+\lambda_{w-1}j_{w-1})$, there are $\lambda_w$ pairwise disjoint $j_w$-sets $A_{w,1},A_{w,2},\cdots, A_{w,\lambda_w}$ sets that are also pairwise disjoint with all the sets in ${\cal A}'$ such that $\sum_{j\in A_{w,i}}M_j=0$ for all $i\in [\lambda_w]$. Let ${\cal A}={\cal A}'\cup\{A_{w,i}|i\in [\lambda_w]\}$. Obviously, $|\cup {\cal A}|=\lambda_1j_1+\cdots+\lambda_wj_w$ and $\sum_{j\in \cup {\cal A}}M_j=0$.
\end{proof}

We now show that if a $k$-good matrix $M$ is $(J,poly(k))$-bad then $\gcd(J)\notdivides k$.

\begin{lemma}\label{kkk} Let $K=\lfloor\sqrt{k}/(2\log k)\rfloor$, $\kappa=k^{1.5} $, $J\subseteq [K]$ and $k/2\le k'\le k$. Let $M$ be a matrix that is $k'$-good and $(J,\kappa )$-bad. Then $\gcd(J)\notdivides k'$.
\end{lemma}
\begin{proof} Let $d=\gcd(J)$ and suppose, for the contrary, that $d|k'$. By Lemma~\ref{N1}, there is $J'\subseteq J$ of size $w:=|J'|\le \log (K/d)<\log k$ such that $d=\gcd(J')$. Let $J'=\{j_1,\ldots,j_w\}$. By Lemma~\ref{Lamb}, there exist $0\le \lambda_1,\ldots,\lambda_{w}\le k$ such that $\lambda_1j_1+\cdots+\lambda_wj_w=k'$. By Lemma~\ref{zero}, $M$ is $k'$-bad. A contradiction. 
\end{proof}

Let $K=\lfloor\sqrt{k}/(2\log k)\rfloor$ and $\kappa=k^{1.5}$. Let $\N_{d}$ be the set of integers in $[K]$ that are not divisible by~$d$.
\begin{lemma}\label{gbad} Let $J$ be the maximum subset of $[K]$ such that $M$ is $(J,\kappa )$-bad. Then $M$ is $(\N_{\gcd(J)},\kappa )$-good.
\end{lemma}
\begin{proof} Since $J$ is the maximum set, $M$ is $([K]\backslash J,\kappa )$-good. Since $J\subseteq [K]\backslash \N_{\gcd(J)}$ we have $[K]\backslash J \supseteq \N_{\gcd(J)}$ and therefore $M$ is $(\N_{\gcd(J)},\kappa )$-good.
\end{proof}

We now show how to construct from a $(\N_d,\kappa )$-good matrix a $(\N_{d',}\kappa )$-good matrix with $d'\ge 2d$.
\begin{lemma}\label{recur} Let $M$ be a $q\times n$ matrix that is $(\N_d,\kappa )$-good. There exist $k'\le k$, $q'= q+\Pi(k',n)$, $d'\ge 2d$ and a $q'\times n$ matrix $M'$ that is $(\N_{d'},\kappa )$-good.
\end{lemma}
\begin{proof} Consider $k'=d\lfloor k/d\rfloor$ and let $\hat M$ be a $\Pi(n,k')\times n$ matrix that is $k'$-good. Let $J'$ be the maximum subset of $[K]$ such that $\hat M$ is $(J',\kappa )$-bad. By Lemma~\ref{kkk}, $\gcd(J')\notdivides k'=d\lfloor k/d\rfloor$ and therefore $\gcd(J')\notdivides d$. By Lemma~\ref{gbad}, $\hat M$ is $(\N_{\gcd(J')},\kappa )$-good. Define $M'=M\circ \hat M$.

First, the number of rows of $M'$ is $q'=q+\Pi(k',n)$. Now, by Lemma~\ref{concatenation}, $M'$ is $(\N_{\gcd(J')}\cup \N_d,\kappa )$-good. Since $\N_{\gcd(J')}\cup \N_d=\N_{d'}$ for $d'=\lcm(\gcd(J'),d)$ we have that $M'$ is $(\N_{d'},\kappa )$-good. Since $\gcd(J')\notdivides d$ we have $d'=\lcm(\gcd(J'),d)\ge 2d$. This implies the result.
\end{proof}

We are ready now to prove the final result
\begin{lemma} For $n\ge k^{2.5}$ there is $k'\le k$ such that $\Pi(n,k')=\Omega((\sqrt{k}/\log^2 k)\log n)$.
\end{lemma}
\begin{proof} Let $M$ be the $1\times n$ matrix $[111\cdots 1]$. Then $M$ is $\N_2$-good. By Lemma~\ref{recur}, there exist $k_1,k_2,\cdots,k_t$, $t=O(\log k)$, $q_t= 1+\Pi(k_1,n)+\cdots+\Pi(k_t,n)$, $d'\ge 2^{t+1}> K$ and a $q_t\times n$ matrix $M'$ that is $(\N_{d'},\kappa )$-good. Since $d'>K$, $M'$ is $([K],\kappa )$-good. By Lemma~\ref{llbb},
$$q_t=\Omega\left(K \log \frac{n- \kappa  K^2}{K}\right)=\Omega\left(\frac{\sqrt{k}}{\log k} \log n\right). $$
Therefore, there exists $k':=k_i$ such that
$$\Pi(n,k')=\Omega\left(\frac{\sqrt{k}}{\log^2 k} \log n\right). $$
\end{proof}

\bibliography{TestingRef}

\end{document}